\documentclass[copyright,creativecommons]{eptcs}
\providecommand{\event}{FROM 2024} 

\usepackage{amsmath,epsf,graphicx,xcolor,amsthm,framed}
\definecolor{shadecolor}{gray}{0.9}
\usepackage{iftex}
\newcommand{\keyw}[1]{{\bf #1}}

\newcommand{\dahntab}[1]{
  \newbox\mybok%
  \setbox\mybok=\hbox{\vbox{
      \begin{tabbing}
        #1
      \end{tabbing}%
    }}

  \newdimen\bokwidth%
  \bokwidth=\wd\mybok%
  \newdimen\myl%
  \myl=\textwidth%
  \divide\myl by 2%
  \divide\bokwidth by -2%
  \advance\myl by\bokwidth%
  \vrule width\myl height 0pt depth 0pt%
  \usebox\mybok%
}

\def\OR{\vee}
\def\AND{\wedge}
\def\goesto{\rightarrow}

\newtheorem{definition}{Definition}
\newtheorem{theorem}{Theorem}

\newtheorem{example}{Example} 
\newtheorem{observation}{Observation}

\title{Towards Geometry-Preserving Reductions Between Constraint Satisfaction Problems (and other problems in NP)}
\author{Gabriel Istrate
\institute{University of Bucharest}
\institute{Str. Academiei 14, Sector 6, 011014, 
Bucharest, Romania}
\email{gabriel.istrate@unibuc.ro}
}

\newcommand{\titlerunning}{Geometry-Preserving Reductions Between CSP (and other problems in NP)}
\newcommand{\authorrunning}{G. Istrate}

\begin{document}
\maketitle 
\begin{abstract}
We define two kinds of \textit{ geometry-preserving reductions} between constraint satisfaction problems and other NP-search problems. We give a couple of examples and counterexamples for these reductions.  
\end{abstract}

\section{Introduction}

Reductions are the fundamental tool of computational complexity. They allow
classification of computational problems into families (somewhat resembling those in biology), with problems in the same class sharing common 
features (in particular NP-complete problems being
isomorphic versions of a single problem, if we believe that the  Berman-Hartmanis conjecture \cite{ber-har:j:iso} holds).

A significant concern in the literature on complexity-theoretic reductions is to make the theory of reductions more predictive. This means that we should aim to better connect the structural properties of the two problems that the given reduction is connecting. For instance the notion of \textit{pasimonious reduction} attempts to preserve the total number 
of solutions between instances. Various attempts (e.g. \cite{creignou1995class,agrawal2001reducing}) have proposed restricting the computational power allowed for computing  reductions. (Strongly) local reductions \cite{local-reductions-tr,hunt2001strongly} provide a principled approach that \textit{ simultaneously} relates the complexity of various versions  (decision, counting,etc) of the two problems
in the reduction. In another direction, several logic-based approaches to the complexity of reductions were considered \cite{stewart1994completeness,dalmau2023local}. 
Finally, motivated by the connection with local search, a line of research \cite{fischer-ipl}, \cite{wi-local-search} 
has investigated reductions that attempt to preserve the \textit{structure} of solution spaces, by not only translating the instances of the two decision problems but also the set of \emph{solutions}. 

The geometric structure of solution spaces of combinatorial problems has been brought
to attention by an entirely different field of research:
that of \textit{phase transitions in combinatorial optimization}  \cite{comp-statphys,mezard-montanari}. It was shown that the predictions of the so-called \textit{one-level replica symmetry breaking} in Statistical Physics applies to (and offer predictions for) the complexity of many constraint satisfaction problems. It is believed \cite{krzakala2007gibbs} that the set of solutions of a constraint satisfaction problem such as $k$-SAT undergoes a sequence of phase transitions, culminating in the SAT/UNSAT transition. The first transition is a \textit{clustering transition} where the set of satisfying assignments of a typical random formula changes from a single, connected cluster to an exponential number of clusters that are far apart from eachother. Some aspects of this solution clustering have been  confirmed rigorously (\textit{solution clustering} \cite{achlioptas-frozen,cond-mat/0504070/prl,mezard-mora,clusters-overlaps,holes-1-in-k,achlioptas2015solution}). The 1-RSB prediction motivated the development of the celebrated \textit{ survey propagation (SP) algorithm} \cite{surveyprop-rsa2005} that predicts (for large enough $k$) the location of the phase transition of the random $k$-SAT problem \cite{ding2022proof}. This breakthrough has inspired significant recent progress \cite{coja2020replica,gamarnik2021overlap,gamarnik2023algorithmic,bresler2022algorithmic}.  

A major weakness of the theory of phase transition is the lack of general results that connect the qualitative features such as clustering and 1-RSB to computational complexity. In contrast, Computational Complexity has developed algebraic classification tools that (at least for the well-behaved class of Constraint Satisfaction Problems) provide an understanding of the structural reasons for the (in)tractability of various CSP \cite{schaefer-dich,bulatov2017dichotomy,zhuk2020proof}. Note, however, that there are some classification results in the area of Phase Transitions: For instance, one can classify 
the existence of thresholds for random Constraint Satisfaction Problems in a manner somewhat reminiscent of the Schaefer Dichotomy Theorem \cite{molloy-stoc2002,hatami-molloy,istrate:ccc00,creignou-daude-threshold,istrate-sharp}. Intuitively, SAT problems without a sharp threshold qualitatively resemble Horn SAT, a problem whose (coarse) threshold is well understood \cite{istrate:horn,istrate-khorn,continuous-discontinuous-journal}. Also, versions of SAT that have a "second-order phase transition" \cite{2+p:nature} seem qualitatively similar to 2-SAT \cite{aimath04}. An example of such a problem which is NP-complete but "2-SAT-like" is 1-in-$K$ SAT, $k\geq 3$ \cite{1-in-k}. Finally, the existence of overlap discontinuities can be approached in a uniform manner for a large class of random CSP \cite{clusters-overlaps}. In spite of all this, the existence of general classification results does \textbf{not} extend to properties (such as 1-RSB) that concern the geometric nature of the solution space of random CSP. 
 
In this paper \textbf{we propose connecting these two directions by
the definition of reductions that take into account the geometric structure of solution spaces}. We put forward two such notions: 
\begin{itemize} 
\item[-] the first type of reductions extends the witness-isomorphic reductions of \cite{wi-local-search} by requiring that the solution overlap of the target (harder) problem depends predictably on the overlap of the source (easier) problem. In other words, the cluster structure of the harder problem is essentially the same as the one of the easier problem.
\item[-] for a subclass of NP-search problems, that of constraint satisfaction problems, we define a type of witness-preserving reductions based on the notion of \textit{ covers}. Covers \cite{maneva-sp,kroc-sp-uai07} were defined in the course of analyzing the seminal \textit{ survey propagation} algorithm of Zecchina et al. \cite{surveyprop-rsa2005}. They are generalizations of (sets of) satisfying assignments for the problem, and allow a geometric interpretation of this algorithm. 
Specifically, survey propagation is just belief propagation applied to the set of covers \cite{maneva-sp}.  Our reduction extend the usual notion of reductions by requiring that not only solutions of one problem are mapped to solutions, but this is true \textit{ for covers as well}. 
\end{itemize} 

The outline of the paper is as follows: in Section 2 we subject the reader to a (rather heavy) barrage of definitions and concepts. The two kinds of reductions we propose are given in Definitions~\ref{op-reduction} and~\ref{cover-reduction}. In Section 3 we give four examples of problem reductions (of various naturalness and difficulty) that belong to at least one of the two classes of reductions that we put forward. The first two are classic examples from the literature on NP-completeness. The next two are slightly less trivial examples, and are constructed to be appropriate for one of the two types of reductions (one for each type). Correspondingly, in Section 4 we point out that the natural reduction from 4-SAT to 3-SAT belongs to neither class. An algebraic perspective on the theory of reductions is called for in Section 5. We end (Section 6) with a number of musings concerning the merits and pitfalls of our approach, and the road ahead. 

\begin{shaded} 
We stress the fact that \textbf{this is primarily a conceptual paper}: its added value consists primarily, we believe, in pointing to the right direction in a problem where nailing the correct concepts is tricky, thus stimulating further discussion, rather than in coming up with completely appropriate concepts and proving things about them. \textbf{Our main goal is to understand and model, not to prove.} However, if "right" definitions exist, they should be compatible with our results. 
\end{shaded} 
\section{Preliminaries}

We review some notions that we will need in the sequel: 

\begin{definition} 
An \textit{NP-search problem} is specified by
\begin{itemize} 
\item a polynomial time computable predicate $A\subseteq \Sigma^{*}$. $A$ specifies the set of well-formed instances. 
\item a polynomial time computable relation $R\subseteq A \times \Sigma^{*}$ with the property that there exists a polynomial-time computable (polynomially bounded) function $q(\cdot)$ such that for all pairs $(x,y)\in A\times \Sigma^{*}$ such that $R(x,y)$ ($y$ is called \textit{a witness for }$x$, and $x$ is called \textit{a positive instance}), we have $|y| = q(|x|)$\footnote{normally we require only the inequality $|y| \leq q(|x|)$, but we pad witnesses $y$, if needed, to ensure that this condition holds.}. 
\item a polynomial time computable relation $N\subseteq A\times \Sigma^{*}\times \Sigma^{*}$ such that if $N(x,y_1,y_2)$ then $|y_1|=|y_2|=q(|x|)$. Intuitively, $N$ encodes the fact that $y_1,y_2$ are "neighbors". Note that it is \textbf{not} required that $R(x,y_1),R(x,y_2)$, i.e. that $y_1,y_2$ are witnesses for $x$. Instead, $y_1,y_2$ are   "candidate witnesses". 
\end{itemize} 

We abuse language and write $R$ instead of $(A,R,N)$. 
We denote by $NP_{S}$ the class of $NP$-search problems. A problem $R$ in $NP_{S}$ belongs to the class $P_{S}$ iff there exists a polynomial time algorithm $W$ s.t.
\begin{itemize} 
\item For all $x\in A$, if $W(x)="NO"$ then for all $y\in \Sigma^{*}$, $(x,y)\not\in R$. 
\item For all $x\in A$, if $W(x)=y\neq "NO"$ then $R(x,y)$. 
\end{itemize} 
\end{definition} 

\begin{definition} Given $NP$-search problem $R$ and positive instance $x$, define \textit{the overlap of two witnesses $y_1,y_2$ for $x$} as 
\begin{equation} 
overlap(y_1,y_2)=\frac{d_{H}(y_1,y_2)}{q(|x|)} 
\end{equation} 
In the previous equation $d_{H}$ denotes, of course, the Hamming distance of two strings. 
\end{definition} 

We will employ a particular kind of NP-search problem that arises from optimization. The particular class of interest consists of  those optimization problems for which determining if a candidate witness is a local optimum,
and finding a better candidate witness otherwise, is tractable:  

\begin{definition}
An NP-search problem $R$ belongs to the class PLS (\textit{polynomial local search}, see e.g. \cite{yannakakis1997computational}) if there exist three polynomial time algorithms $A$, $B$, $C$ such that 
\begin{itemize} 
\item Given $x\in \Sigma^{*}$, $A$ determines if $x$ is a legal instance (as in Definition 1), and, additionally, if this is indeed the case, it produces some candidate witness $s_0\in \Sigma^{q(|x|)}$. 
\item Given $x\in A$ and string $s$, $B$ computes a cost $f(x,s)\in \mathbf{N}\cup \{\infty\}$\footnote{Rather than adding another relation to test whether $s$ is a legal witness candidate for $x$, we chose to do this implicitly, by requiring that $f(x,s)<\infty$}. $s$ is called a \textit{local minimum} if $f(x,s)\leq f(x,s^{\prime})$ for all $s^{\prime}$ such that $N(x,s,s^{\prime})$. Local maxima are defined similarly. 
\item Given $x\in A$ and string $s$, $C$ determines whether $s$ is a local optimum and, if it not, it outputs a string $s^{\prime}$ with $N(x,s,s^{\prime})$ such that $f(x,s^{\prime})<f(x,s)$ (if $R$ is a minimization problem; the relation is reversed if $R$ is a maximization problem).  
\end{itemize} 
In addition we require that $R(x,s)$ (i.e. $s$ is a witness for $x$) if and only if $s$ is a local optimum for $x$. That is, witnesses for an instance $x$ are local optima. 
\end{definition} 

\begin{observation} 
It is known (see e.g. \cite{yannakakis1997computational}) that $P_{S}\subseteq PLS\subseteq NP_{S}$. 
\end{observation} 

\begin{example} 
MAX-SAT is the following optimization problem: an instance of MAX-SAT is a CNF formula $\Phi$ whose clauses $C_1,C_{2},\ldots, C_{m}$ are endowed with a positive weight $w_1,w_2,\ldots, w_m$, respectively. The set of witnesses for $\Phi$ is the set of  assignments of the variables in $\Phi$. Each assignment $A$ comes with a weight $w(A)$, defined to be the sum of weights of all clauses satisfied by $A$. 

A candidate witness for $\Phi$ is a truth assignment of the variables in $\Phi$. Two candidate witnesses $w_1,w_2$ are neighbors if they differ on the value of a single variable. 

A witness for $\Phi$ is an assignment $w(A)$ which is a local maximum, i.e for all assignments $B$ differing from $A$ in exactly one position $w(B)\leq w(A)$. 

It is clear that (as defined) problem MAX-SAT belongs to class $PLS$. 
\end{example} 

\begin{definition} 
Given problem $A\in PLS$, define the NP-search problem $A^{*}$ as follows:
\begin{itemize} 
\item Inputs of $A^{*}$ consist of pairs $(x,\tau,0^{d})$, where $x$ is an input for $A$, $\tau$ is a string in $\Sigma^{*}$ for some finite alphabet $\Sigma$, $|\tau|\leq q(|x|)$, and $d\geq 1$ is an integer. 
\item A witness for pair $(x,\tau,0^{d})$ is a path $\tau_{0},\tau_{1},\ldots, \tau_{d^{\prime}}$ with $d^{\prime}\leq d$ such that (a). $\tau_{0}=\tau$ (b). for every $i\geq 1$ $\tau_{i-1}$ and $\tau_{i}$ are neighbors and $f(x,\tau_{i-1})<f(x,\tau_{i})$. 
(c). $\tau_{d^{\prime}}$ is a local optimum in $A$ 
(such a path is called \textit{an augmenting path}). 
\end{itemize}  
In other words, given instance $(x,\tau,0^{d})$ the problem $A^{*}$ is to decide whether there exists an augmenting path of length at most $d$ from $\tau$ to a local optimum. 
\label{def-a-star}
\end{definition} 

To define cover-preserving reductions we first need  \textbf{ witness-isomorphic reductions}, a
very restrictive notion that has been previously studied in the literature \cite{wi-local-search}. 

\begin{definition} 
Given $NP$-decision problems $(A,R_A)$ and $(B,R_{B})$ a \textit{witness-isomomorphic reduction of $(A,R_A)$ to $(B,R_B)$} is specified by two polynomial time computable and polynomial time invertible functions $f,g$ with the following properties: 
\begin{itemize} 
\item[-] for all $x$, $x\in A\Leftrightarrow f(x)\in B$. That is, $f$ witnesses that $A\leq_{m}^{P} B$. 
\item[-] for all $x,y$ if $R_{A}(x,y)$ then there exists a $z$ such that $g(<x,y>)=<f(x),z>$ and $R_{B}(f(x),z)$. 
\item[-] for all distinct $y_1,y_2$ if $R(x,y_1)$ and $R(x,y_2)$ then $g(<x,y_1>)\neq g(<x,y_2>)$. 
\item[-] for all $x,w,z$ if $f(x)=w$ and $R_{B}(w,z)$ then there exists $y$ such that $R_{A}(x,y)$ and $g(<x,y>)=<w,z>$. 
\end{itemize} 
We will write $(A,R_A)\leq_{wi} (B,R_B)$ when there exists such a witness isomorphic reduction. 
\end{definition} 

A particular well-behaved type of NP-search problem is the class of Constraint Satisfaction Problems: 

\begin{definition}\label{csp}
The \textit{ Constraint Satisfaction Problem} $CSP(C)$ is specified by
\begin{enumerate}
 \item A finite \textit{ domain}. Without loss of generality we will consider CSP with domain  $D_{k}=\{0,\ldots, k-1\}$, for some $k\geq 2$.
\item A finite set $C$ of \textit{ templates}. A template is a finite subset $A\subseteq D_{k}^{(r)}$ for some $r\geq 1$, called \textit{ the 
arity of template $A$}.
\end{enumerate}

An instance $\Phi$ of $CSP(C)$ is a formula obtained as a conjunction of instantiations of templates from $C$ to tuples of variables 
from a fixed set $x_{1},\ldots, x_{n}$.
\end{definition}

\begin{definition} \cite{kroc-sp-uai07}
 A \textit{ generalized assignment} for a boolean CNF formula is a mapping $\sigma$ from the variables of $F$ to $\{0,1,*\}$. 
Given a generalized assignment $\sigma$ for $F$,
variable $x$ is called \textit{ a supported variable under $\sigma$} if there exists a clause $C$ of $F$ such that $x$ is the only variable of $C$ that is satisfied by $\sigma$, and all the other literals are assigned FALSE.
\end{definition}

For CSP over general domains the previous definitions need to be generalized:

\begin{definition} (see also \cite{tu2008survey})
 A \textit{ generalized assignment} for an instance of $CSP(C)$ over domain $D_{k}=\{0,1,\ldots, k-1\}$ is a mapping $\sigma$ from 
the variables of $\Phi$ to ${\mathcal P}(D_{k})\setminus \{ \emptyset \}$. $\sigma$ is simply called
an \textit{ assignment} if $|\sigma(v)|=1$ for all variables $v$ of $\Phi$. Given assignment $\sigma_{1}$ and generalized assignment 
$\sigma_{2}$, we say that \textit{ $\sigma_{2}$ is compatible with $\sigma_{1}$} (in symbols $\sigma_{1}\subseteq \sigma_{2}$) if for every variable $v$,  $\sigma_{1}(v)\subseteq \sigma_{2}(v)$.
\end{definition}

For propositional formulas in CNF the important notions of \textit{ supported variable} and \textit{ cover} were introduced in \cite{maneva-sp},  
\cite{kroc-sp-uai07}, \cite{jsat-support}:

\begin{definition} \label{sup}
Assignment $\sigma\in \{0,1,*\}^{*}$  is a \textit{ cover} of $F$ iff:
\begin{enumerate}
 \item every clause of $F$ has at least one satisfying literal or at least two literals with value $*$ under $\sigma$ (``no unit 
propagation''), and
\item $\sigma$ has no unsupported variables assigned 0 or 1.
\end{enumerate}
It is called a \emph{true cover} (e.g. \cite{kroc-sp-uai07}) if there exists a satisfying assignment $\tau$ of $F$ such that $\tau \subseteq \sigma$.  In this paper we will only deal with true covers.   
\end{definition}

One extension of the notion of supported variable to general CSP could be:

\begin{definition} \label{support-modified}
 Given a generalized assignment $\sigma$ for an instance $F$ of $CSP(C)$, variable $x$ is called \textit{ supported by $\sigma$} if there 
exists a constraint $C$ of $F$ such that
\begin{enumerate}
 \item $|\sigma(v)|=1$ for every $v\in C$.
\item for every $\lambda\neq \sigma(x)$, changing the value of $\sigma(x)$ to $\lambda$ results in a generalized assignment that does 
not satisfy $C$.
\end{enumerate}
\end{definition}

As for the extension of the notion of cover to general CSP, we have to modify Definition~\ref{sup} for two reasons: first, we need more than three symbols to encode all the $2^{k}-1$ possible choices in assigning a variable. Second, in the general case constraints cannot be satisfied by setting a single variable.

\begin{definition} A generalized assignment $\sigma\in D_{k}^n$ of a CSP $F$ is
a \textit{ cover} of $F$ iff:
\begin{itemize}
 \item no constraint of $F$ can further eliminate any values of $\sigma$ for its variables by consistency (``no unit propagation''; see also \cite{tu2010sp}).
\item $\sigma$ has no unsupported variables $v$ (in the sense of Definition~\ref{support-modified}) s.t. $|\sigma(v)|=1$.
\end{itemize}
\end{definition}

\begin{definition} 
\label{op-reduction} 
Given $NP$-decision problems $(A,R_A)$ and $(B,R_B)$, a witness-isomorphic reduction of $(A,R_A)$ to $(B,R_B)$ is called \textbf{overlap preserving} if there is a continuous, monotonically increasing function $h:[0,1]\rightarrow [0,1]$ with $h(0)=0,h(1)=1$ such that for all $x$, $|x|=n$, and $z_1\neq z_2$ such that $R_{A}(x,z_1)$ and $R_{A}(x,z_2)$, 
\[
overlap(\pi_{2}^{2}\circ g(x,z_1), \pi_{2}^{2}\circ g(x,z_2))=h(overlap(z_1,z_2))+o_{n}(1). 
\]
\end{definition} 

In other words given two solutions $z_1,z_2$ of the instance $x$ of $A$, one can map them (via the reduction) to two solutions $w_1,w_2$ of instance $f(x)$ of $B$ such that the overlap of $w_1,w_2$ depends predictably on the overlap of the original instances $z_1,z_2$. In particular if $x$ has a single cluster of solutions with a single overlap (respectively many clusters of solutions with a discontinuous overlap distribution) as predicted by the 1-RSB ansatz for random $k$-SAT) then the solution space of $f(x)$ should have a similar geometry.

\begin{definition}
\label{cover-reduction}
 A wi-reduction $(f,g)$ between search problems $A$ and $B$ is called \textbf{ cover-preserving} if there exists a polynomial-time computable 
mapping $\overline{g}$ such that
\begin{enumerate}
 \item $\overline{g}$ extends $g$, i.e. if $x$ is an instance and $y$ is a witness for $x$ then $\overline{g}(<x,y>)=g(<x,y>)$. 
\item $(f,\overline{g})$ is witness isomorphic when seen as a reduction \textbf{between covers.} 
\item $\overline{g}$ is compatible with $g$. That is, for all $x\in A$, for all $y$ witnesses for $x$, for all generalized assignments 
$z$ compatible with $y$, if $z$ is a cover of $x$ then $\overline{g}(<x,z>)$ is a cover of $f(x)$ compatible with $g(<x,y>)$.
\end{enumerate}
\end{definition}

We will informally refer to reductions that have both properties as \textit{geometry preserving}.

\section{Geometry-preserving reductions via local constructions}

Our first example of a geometry-preserving reduction is the classical reduction from  $k$-coloring to $k$-SAT. The reduction is the usual 
one:
to each vertex we associate three logical variables $x_{i,1}, x_{i,2},x_{i,3}$.
To each edge $e=(v_{i},v_{j})$ we associate a formula $F_{e}=F_{e}(x_{i,1},x_{i,2},x_{i,3},x_{j,1},x_{j,2},x_{j,3})$. The formula 
$F_{e}$ codifies three types of constraints: each vertex is given at least one color; no vertex is given more than one color; finally 
adjacent vertices cannot
get the same color. That is
\begin{equation} 
 F_{e}=G(x_{i,1},x_{i,2},X_{i,3})\AND G(x_{j,1},x_{j,2},x_{j,3}))\AND (\AND _{k\neq l \in \{1,2,3\}} H(x_{i,1},x_{i,2})\AND H(x_{j,k},x_{j,l}))
\end{equation}

\begin{theorem}\label{red1}
The natural reduction from $k$-COL to $k$-SAT is cover-preserving and overlap-preserving. 
\end{theorem}
\begin{proof} 
 Let $G$ be an instance of $k$-COL and $\Phi_{G}$ the instance of $k$-SAT that is the image of $G$ under the natural reduction. 
Let $c\in [k]^{|V(G)|}$ be a coloring of $G$ and $w$ be a generalized assignment (a.k.a. list coloring) compatible with $c$. Under 
the natural transformation each vertex $v$ is represented by $k$ logical variables, $x_{v,1},x_{v,2},\ldots, x_{v,k}$, and the 
constructed formula ensures that only one of them can be true. In particular one can define $f(c)$ to be the satisfying assignment 
naturally corresponding to the coloring $c$.  

Now transform a general assignment $w$ of colors to the graph $G$ to a generalized assignment $g(G,w)$ for the formula by the following rules: 
\begin{enumerate} 
 \item if $\emptyset \neq C_{v}\subseteq \{1,2,\ldots, k\}$ is a set of allowed colors for vertex $v$ then we set all variables 
$x_{v,i}$, $i\not \in C_{v}$, to zero. 
\item if $|C_{v}|=1$ then we set the remaining variable $x_{v,i}$ to 1. If, on the other hand, $|C_{v}|\geq 2$, then we set all 
variables $x_{v,i}$, $i\in C_{v}$ to $*$. 
\end{enumerate} 

It is easy to see that $w$ is a cover compatible with $c$ if and only if $g(<G,w>)$ is a cover compatible with $g(<G,c>)$. 
\end{proof}

\begin{theorem} 
The natural reduction from 1-in-k SAT to $k$-SAT, $k\geq 3$ is overlap preserving and cover preserving. 
\end{theorem} 
\begin{proof} 
Remember, clause 1-in-$3(x,y,z)$ is translated as $x\OR y \OR z, \overline{x}\OR \overline{y}\OR \overline{z}$. Since the reduction does not add any new variables, the set of satisfying assignments is the same, hence the overlap of the translated solutions does not change. Similarly, covers correspond naturally (via identity) to covers. 
\end{proof}

\subsection{An(other) example of an overlap-preserving reduction}

The purpose of this section is to give an example of a somewhat less trivial witness-isomorphic reduction that is overlap-preserving. The reduction is one that was already proved to be witness-isomorphic in the literature. Specifically, Fischer \cite{fischer-ipl} has proved that problem MAX-SAT$^{*}$ (see Definition~\ref{def-a-star}) is NP-complete. Subsequently 
Fischer, Hemaspaandra and Torenvliet \cite{wi-local-search} have shown that the reduction between $SAT$ and $MAX-SAT^{*}$ is witness isomorphic. 

\begin{theorem} One can encode $MAX-SAT^{*}$ as a NP-search problem such that the witness-isomorphic reduction from $(SAT,R_{SAT})$ to (MAX-SAT$^{*},R_{MAX-SAT^{*}})$ \cite{wi-local-search} is overlap-preserving. 
\end{theorem} 
\begin{proof} 
First let's recall the reduction from \cite{wi-local-search}: 
let $\phi$ be a propositional formula in CNF with $n$ variables $x_{1},\ldots, x_{n}$ and $m$ clauses $C_{1},\ldots, C_{m}$. Define
\[
\psi = \bigwedge_{i=1}^{m} (x_{i}\OR b_{i})^{m+1}\AND (C_{1}\OR \alpha)\AND \ldots \AND (C_{m}\OR \alpha)
\]
In this formula $b_{1},\ldots, b_{m}, \alpha$ are new variables $C^{m+1}$ denotes the fact that clause $C$ has weight $(m+1)$, while $C\OR \alpha$ denotes the clause obtained by adding $\alpha$ to the disjunction in $C$. A clause without an upper index has weight 1. Let 
\[
\xi = \bigwedge_{j=1}^{n-1} \bigwedge_{i=j+1}^{n} (x_{j}\OR b_{j}\OR \overline{x_{i}})^{3m}\AND (x_{j}\OR b_{j}\OR \overline{b_{i}})^{3m}. 
\]
Finally, let $\Psi = \psi \AND \xi$.

It is easy to see (and was proved explicitly in \cite{fischer-ipl}) that witnesses for instance $(\Psi,0^{2n+1},0^{n})$ of MAX-SAT$^{*}$ correspond to satisfying assignments $A$ for $\Phi$ as follows: we start at $\tau_{0}=0^{2n+1}$. First we flip either variable $x_1$ (if $A(x_1)=1$) or $b_1$ (if $A(x_1)=0$). Next we flip either variable $x_2$ (if $A(x_2)=1$) or $b_2$ (if $A(x_2)=0$). $\ldots$. Finally, we flip either variable $x_n$ (if $A(x_n)=1$) or $b_n$ (if $A(x_n)=0$). Denote by $P_A$ the path obtained in this way.

Our goal is to encode path $P_A$ in such a way that for two 
assignments $A,B$, $d_{H}(P_{A},P_{B})=(2n+2)\cdot d_{H}(A,B)$. Since we encode one bit of a satisfying assignment by $(2n+2)$ bits of the encoding of paths, this relation proves that the reduction is overlap preserving. Indeed, $overlap(A,B)=\frac{d_{H}(A,B)}{n}$, and $overlap(P_{A},P_{B})=\frac{(2n+2)\cdot d_{H}(A,B)}{n(2n+2)}=overlap(A,B)$, so we can take $h(x)=x$ in Definition~\ref{op-reduction}. 

The idea of the encoding is simple: we use strings $z_1z_2\ldots z_{d}$, where each $z_i$ has length $2n+2$ and contains $n+1$ consecutive ones (including circular wrapping). The idea is that we encode into string $z_i$ the event of flipping the value of the $j$'th variable in a fixed ordering of variables of $\Phi$ (specified below) by 
making $z_i$ consist of $n+1$ consecutive ones, starting at position $j$. To make the encoding preserve the overlap predictably, make the position of variable $b_j$ differ by $n+1$ (modulo $2n+2$) from the position of variable $x_j$. 

Consider now two satisfying assignments $A$ and $B$. If $A$ and $B$ agree on whether to flip $x_i$ or $b_i$ then the corresponding $i$'th blocks of $P_A,P_B$ agree. Otherwise the corresponding blocks of $P_A,P_B$ are complementary. So indeed we have $d_{H}(P_{A},P_{B})=(2n+2)\cdot d_{H}(A,B)$. 
\end{proof}

\subsection{An(other) example of a cover-preserving reduction}

Whereas in the previous section we gave a slightly less trivial reduction that was overlap-preserving, the goal of this subsection is to  give a slightly less trivial example of a cover-preserving reduction. The problem we are dealing with will 
not fall in the framework of Definition~\ref{csp} (and subsequent Definitions 6--12) since it also includes a ``global'' constraint. Rather than attempting to rewrite many of the previous definitions in order to construct a more  
general framework, appropriate to our example, we will be content to keep things at an intuitive level, adapting the definitions on the spot as needed. We trace the idea of the construction below to an unpublished, unfinished manuscript \cite{bp-bisection} that has no realistic chances of ever being completed, as a (different) analysis of the belief propagation algorithm for graph bisection using methods from Statistical Physics has been published in the meantime \cite{vsulc2010belief}. The adaptation of the construction to our (rather different) purposes and the statement/proof of the theorem below is, however, entirely ours. 

\begin{definition}
 The \textit{ Perfect graph bisection (PGB) problem} is specified as a search problem as follows:
\begin{itemize}
 \item{\bf Input:} A graph $G=(V,E)$.
 \item{\bf Witnesses:} A function $f:V\goesto \{-1,1\}$ such that for all vertices $v,w\in V$, $(v,w)\in E\Rightarrow f(v)=f(w)$ and $\sum_{v\in V} f(v)=0$. 
 \end{itemize} 
 We identify solutions $f,g$ such that $f(v)=-g(v)$ for all $v\in V$.
  \end{definition} 
 \begin{definition} 
 To define PGB as a local search problem we will work with partitions, defined as functions $f:V\goesto \{-1,1\}$.  Two partitions are \textit{ adjacent} if one can be obtained from the other by flipping two vertices on opposite sides of the partition. 
Covers are defined as strings over $\{-1,1,*\}$, where again we identify strings that only differ by permuting labels -1 and 1. 
\end{definition}

The first step in order to set up PGB as a local search problem
subject to cover-preserving reductions is to
represent $G$ as a \textit{ factor graph\/}~\cite{factor}.  A factor graph is a
bipartitite undirected graph containing two different kinds of nodes:
\textit{ variable nodes\/}, corresponding to
vertices $v_i$ in $G$, and \textit{ function nodes\/} corresponding to
edges $e_{ij}$ in $G$.  A variable node and function node are connected
in the factor graph if the edge represented by the function node is
incident on the vertex represented by the variable node. It is clear that a function node is adjacent to exactly two variable nodes. 

The constraint on edges only connecting vertices in the same subset
is enforced in an entirely straightforward and local manner in the
factor graph: a function node requires
both its neighboring variable nodes to be assigned the same value $s=+1$ or
$-1$.  The balance constraint $\sum_{i=1}^n s_i = 0$, however, remains a
global one.

Therefore, the next step is to give a 
representation of perfect graph bipartition that allows balance to be enforced in a purely local manner.  

\begin{definition} Given $r$ variables that can take 
values $\{-1,0,1\}$, the \emph{counting constraint $CC_{m}(x_{1},\ldots, x_{r})$} is true if and only if $\sum_{i} x_{i}\in \{\pm m\}$. 

\end{definition} 

\begin{definition} 
 An instance $\Psi$ of a \textit{ counting CSP problem} consists of $n$ variables connected by a set of counting constraints $CC_{n-4}$. The variables can take 
values $\{-1,0,1\}$. A witness for the counting CSP $\Psi$ is an assignment of values to variables such that all constraints are satisfied. 
\end{definition}

\begin{theorem}
There exists a cover-preserving, overlap-preserving reduction from PGB to counting CSP.
\end{theorem}
\begin{proof} 

Consider an instance of $G$ of GBP as described above. 
We translate it into an instance $\Psi_{G}$ of counting CSP by constructing a new factor graph with ${n
\choose 2}$ variable nodes $v_{kl}$, $1\le k<l \le n$.  Whereas a
function node $e_{ij}$ was previously connected to 2 variable nodes
$v_i$ and $v_j$, it is now connected to $2n-4$ variable nodes.  Taking
the case where $2<i<j-1<n-2$, they are:
$v_{1,i},v_{2,i},\dots,v_{i-1,i},v_{i,i+1},\dots$, $v_{i,j-1}$, $v_{i,j+1},
\dots,v_{i,n-1},v_{in}$ and $v_{1j},v_{2j},\dots,v_{i-1,j},v_{i+1,j},
\dots,v_{j-1,j}$, $v_{j,j+1},\dots$, $v_{j,n-1},v_{j,n}$.  Other cases ($j>i$, etc.) may be treated
similarly: only the labeling details change.  Note that variable node
$v_{ij}$ is \textit{ not\/} connected to function node $e_{ij}$.  The
reasons for this will become apparent shortly.

In our new factor graph, assign to each of the ${n \choose 2}$ variable
nodes $v_{kl}$ a value
$x_{kl}=(s_k+s_l)/2$.  Note that $x_{kl}\in\{-1,0,+1\}$, depending on
whether $s_k=s_l=-1$, $s_k\neq s_l$, or $s_k=s_l=+1$.  Next, we will
show how the graph bisection constraints translate into a local
constraint on the values of $x_{kl}$ for all variable nodes connected to
a given function node $e_{ij}$.

We start by providing local conditions in the new factor graph that are
necessary for the perfect graph bisection constraints to be satisfied.  We then
show that these conditions are not only necessary but also sufficient.

Consider the sum of the values $x_{kl}$ for all $2n-4$ variable nodes
connected to a function node $e_{ij}$:

\begin{eqnarray*}
\sum_{\stackrel{k,l\, :\}}{v_{kl} \,\sim\, e_{ij}}} x_{kl}
&=& \sum_{m=1}^{i-1}x_{mi} + \sum_{{m=i+1 \atop m\neq j}}^n x_{im}
+ \sum_{{m=1 \atop m\neq i}}^{j-1} x_{mj}
+ \sum_{m=j+1}^n x_{jm} =
 \sum_{{m=1 \atop m\neq i,j}}^n \frac{(2s_m + s_i + s_j)}{2}
= \sum_{m=1}^n s_m + \frac{(n-4)(s_i+s_j)}{2}.
\end{eqnarray*}

For a balanced solution, $\sum_{m=1}^n s_m=0$.  Furthermore, since $v_i$
and $v_j$ are connected by an edge $e_{ij}$, they must be assigned to
the same partition, and so $s_i=s_j$.  Thus, the necessary condition for
the graph bisection constraints to be satisfied is that for all function
nodes $e_{ij}$,

\begin{equation}
\sum_{k,l\, :\, v_{kl} \,\sim\, e_{ij}} x_{kl} = \pm (n-4).
\label{condition}
\end{equation}

The construction of $\Psi_{G}$ is now clear: it is the conjunction of counting constraints specified by equation~(\ref{condition}). We identify two solutions of $\Psi_{G}$ that can be obtained by multiplying all values by $-1$. 

First we have to show that what we constructed is indeed a reduction. That is, 
we first show that except at small $n$ and for a very small number of
graphs, the conjunctions of conditions~(\ref{condition}) is also a sufficient condition for the existence of a perfect bipartition in $G$.  Suppose, indeed that~(\ref{condition}) are satisfied but that the graph bisection constraints are not.  There are three ways this can happen: 1) there is a balanced solution with an edge constraint violated, 2) there is an unbalanced solution with no edge constraints
violated, and 3) there is an unbalanced solution with an edge constraint
violated.

Consider case 1.  If there is a balanced solution with an edge $e_{ij}$
such that $s_i\neq s_j$, $\sum_{k,l\, :\, v_{kl} \,\sim\, e_{ij}} x_{kl} = 0$.
This is clearly inconsistent with (\ref{condition}) unless $n=4$.

Now consider case 2.  The only way that (\ref{condition}) can be
respected when there is an imbalance and no edge constraints violated is
if $\sum_{m=1}^n s_m=\pm 2(n-4)$.  Since $-n \leq \sum_{m=1}^n s_m
\leq +n$, this is impossible unless $n\leq 8$.

Finally, consider case 3.  In this case, (\ref{condition}) could still
be respected if $\sum_{m=1}^n s_m=\pm (n-4)$.  For that to happen,
however, 2 vertices must be assigned to one partition and $n-2$ vertices
to the other, and the \textit{ all\/} edges in the graph must be violated.
In that case, no edges can be contained within either partition.  

Note that solutions $f:V\rightarrow \{-1,1\}$ of $G$ naturally correspond to solutions $x^{(f)}$, $x^{(f)}_{k,l}=\frac{f(k)+f(l)}{2}$, of the corresponding formula $\Psi_{G}$. Now the fact that the reduction is cover-preserving is easy: a cover $c$ for a solution $f$ consists of specifying for some pairs of vertices of $G$ whether they are always on the same side of a perfect bipartition, whether they are always on opposite sides of a perfect bipartition, or whether they can fluctuate between sides.  We construct a cover $\tilde{c}$ for $x^{(f)}$ corresponding to $c$ as follows: 

\begin{itemize} 
\item[-] if two vertices $k,l$ are always on the same side of the partition in any perfect bipartition then the variable $x^{(f)}_{k,l}$ never takes value 0. In this case $\tilde{c}_{k,l}=\mathbf{1}$. 
 \item[-] If they always are on opposite sides then the variable $x_{k,l}$ always takes value 0. So we set $\tilde{c}_{k,l}=\mathbf{0}$. 
 \item[-] Otherwise $x^{(f)}_{k,l}$ may take all three values.  We set $\tilde{c}_{k,l}=\mathbf{*}$. 
 \end{itemize} 
 
 As for overlaps, given two perfect bipartitions $f,g$ of $G$, we define 
\begin{equation} 
overlap(f,g)=\frac{1}{n} min(d_{H}(f,g),d_{H}(-f,g)).
\label{ov-2} 
\end{equation} 
Similarly, given two solutions $x,y$ for $\Psi_{G}$ define 
\begin{equation} 
overlap(x,y)=\frac{1}{{{n}\choose {2}}} min(d_{H}(x,y),d_{H}(-x,y)). 
\label{ov-3}
\end{equation} 
In equations~(\ref{ov-2}) and~(\ref{ov-3}) we take into account, of course, the factoring of the solution spaces by equivalence. The proportionality factor of~(\ref{ov-3}) is $\frac{1}{{{n}\choose {2}}}$ since formula $\Psi_{G}$ has that many variables. 
Now it is easy to see that 
\begin{equation} 
d_{H}(x^{(f)},x^{(g)})=\frac{ {{n\cdot d_{H}(f,g)}\choose 2}}{{{n}\choose {2}}}=(d_{H}(f,g))^2+o_{n}(1)
\label{ham1}
\end{equation} 
and similarly for $-f$ and $g$, so we can take $h(x)=x^2$. Indeed, if $f$ can be made to coincide on $v$ vertices with $g$ then $x^{(f)}$ and $x^{(g)}$ coincide on all pairs $x^{(f)}_{k,l}$ and $x^{(g)}_{k,l}$, where $k,l$ range over the set of variables where $k,l$ coincide. If $f,g$ differ on one of $k,l$ or both then $x^{(f)}_{k,l}\neq x^{(g)}_{k,l}$. So relation~(\ref{ham1}) is clear. 
\end{proof}

\section{Reductions that are not geometry-preserving} 

A basic sanity check is to verify that there are natural reductions between versions of satisfiability that are not geometry preserving. In this section we give two examples: 

\begin{theorem} 
The classical reduction from 4-SAT to 3-SAT is neither cover preserving nor overlap preserving. 
\end{theorem} 
\begin{proof} Remember, the idea of simulating a 4-clause $x\OR y\OR z\OR t$ by a 3-CNF formula is to add an extra variable $\alpha$ and simulate the clause above by the formula 
$x\OR y\OR \alpha, \overline{\alpha}\OR z \OR t$. 

To show that the reduction is not cover preserving it is enough to create a formula such that by the reduction we create more covers. 

We accomplish this as follows: we add to the formula clause $C_1=x\OR y \OR z \OR t$. Then, informally, we enforce constraints $x\OR y$ and $z\OR t$ using 4-clauses. Let us accomplish this for the first one (the second case is entirely analogous) as follows: we add clauses $x\OR y \OR x_1\OR x_2, x\OR y \OR \overline{x_1}\OR x_2, x\OR y \OR x_1\OR \overline{x_2}, x\OR y \OR \overline{x_1}\OR \overline{x_2}$. 
Let $\Phi_1$ be the resulting formula and $\Phi_2$ be its 3-CNF reduction. 

The set of solutions of $\Phi_1$ is basically $\{(x,y)\in \{(0,1),(1,0),(1,1)\}\} \times \{(z,t)\in \{(0,1),(1,0),(1,1)\}\}$. None of the variables $x,y,z,t$ can be given a 0/1 value in a cover since they are not supported in clause $C_1$ (as there are at least two true literals in that clause). So the set of covers of $\Phi_1$ coincides with the set of satisfying assignments of $\Phi$, together with the trivial cover $***\ldots**$. 

On the other hand, in $\Phi_2$ auxiliary variable $\alpha$ can be given value $*$. In particular $x=1,z=1,\alpha=*$ and all other variables given value $*$ is a nontrivial cover. 
As all values for $\alpha$ are good, the reduction is also not witness isomorphic, hence not overlap preserving. 
\end{proof} 

\section{Geometry-preserving reductions: an Algebraic Perspective}

In this section we advocate for the development of an algebraic approach to reductions, perhaps  based on \textit{ category theory} \cite{categories-working-math}. The motivations for advocating such an approach are simple: 
\begin{itemize} 
\item[-] One can view each constraint satisfaction problem $CSP(S)$ 
as a category: its objects are $S$-formulas. Given 
two objects (formulas) $\Phi_{1}$ and $\Phi_{2}$, a morphism is a mapping between variables of $\Phi_{1}$ and 
variables of $\Phi_{2}$ that maps clauses of $\Phi_{1}$ to clauses of $\Phi_{2}$. 
\item[-] CSP's, together with polynomial time reductions form a category $\mathcal{CSP}_{\leq_{m}^{P}}$. 
\item[-] We can view our constructions in Section 3, however simple, in a categorical fashion. The proof of Theorem~\ref{red1} 
shows that the reduction ``localizes'': we can interpret the usual gadget-based reduction (and in particular our cover reduction) as 
``gluing'' local morphisms. Then it is possible to construct cover-preserving reductions by: 
\begin{enumerate} 
 \item Adequately constructing them ``locally''. 
\item ``Gluing'' them in a way that preserves covers. 
\end{enumerate}

\end{itemize}

The following 
result shows that cover-preserving reductions also can act as morphisms: 
\begin{theorem} The following (simple) properties are true: 
\begin{enumerate} 
 \item For any CSP $A$ the identity map is a cover-preserving reduction of $A$ to itself. 
\item Cover-preserving reductions compose: if $A,B,C$ are CSP and $\leq_{A,B}$, $\leq_{B,C}$ are cover-preserving reductions of $A$ to 
$B$ and $B$ to $C$, respectively, then $\leq_{A,C}:= \leq_{B,C}\circ \leq_{A,C}$ is a cover-preserving reduction of $A$ to $C$. 
\end{enumerate}

Consequently, constraint satisfaction problems with cover-preserving reductions as morphisms form a category $\mathcal{CSP}_{cover}$, 
indeed a subcategory of $\mathcal{CSP}_{\leq_{m}^{P}}$.

Similar results hold for overlap-preserving reductions. 
\end{theorem}
\begin{proof} 
 \begin{enumerate} 
  \item This is trivial. 
\item Let $x$ be an instance of $A$, $y$ be the instance of $B$ that corresponds to $x$ via $\leq_{A,B}$ and $z$ be 
the instance of $C$ that corresponds to $y$ via $\leq_{B,C}$. Further, let $a,a_{1}$ be strings. $a$ is a cover of a witness $a_{1}$ 
for $x$, if and only if $b:=\leq_{A,B}(a)$ is a cover of witness $b_{1}=\leq_{A,B}(a_{1})$ of instance $y=\leq_{A,B}(x)$. This last 
statement happens if and only if $c:= \leq_{B,C}(b)$ is a cover of witness $c_{1}:= \leq_{B,C}(b_{1})$ of instance $z= \leq_{B,C}(y)$. 

But $z=\leq_{B,C}\circ \leq_{A,B}(x)$, $c_{1}=\leq_{B,C}\circ \leq_{A,B}(a_{1})$ and $c=\leq_{B,C}\circ \leq_{A,B}(a)$, so the 
composition of reductions is a reduction and cover preserving. It is clearly specified by a polynomial computable function.  
 \end{enumerate}
 
 Similar arguments function for overlap-preserving reductions. Here we use the fact that the composition of monotonically increasing, continuous functions from [0,1] to [0,1] (such that h(0)=0, h(1)=1) has the same properties. 
\end{proof}

Note that there have recently been several attempts to consider ordinary reductions in computational complexity theory from a structural perspective  \cite{mazzafunctorial,adam2022sheaf,mazza2022categorical}. What can be done for ordinary reductions might be interesting, we believe, for (variants of) our reductions as well.

\section{Conclusions}

Our research clearly has a preliminary character: we see our contribution as a conceptual one, pointing in the right direction rather than having provided the definitive definitions for structure-preserving reductions. Specifically, we believe that one needs to be serious about connecting phase transitions in Combinatorial Optimization to Computational Complexity Theory, however weak the resulting notions may turn out to be. 
(\cite{durgin2020csp} seems to be an intriguing approach in this direction). 

It may be that the results in this paper can be redone using a more restrictive notion of reductions that turns out to be better suited. Whether different variations on the concepts are more appropriate is left for future research. Also left for further research is the issue of classifying the natural encodings of constraint satisfaction problems (particularly for NP-complete problems) under these definitions. This might be difficult, though: our result only shows, for instance, that the natural reduction from 4-SAT to 3-SAT does not have the desired properties, \textit{not that there is no such reduction !}. Also, one cannot solve this problem, in principle without understanding the nature of the phase transition in 3-SAT: the 1-RSB approach predicts the phase transition of $k$-SAT only for large enough values of $k$, and it is not clear that it does for $k=3$. 

In any case we expect that NP-complete problems split in multiple degrees under (variants of) our reductions. Reasons for this belief are twofold: 
\begin{itemize} 
\item[-] First, it is actually known that not all NP-complete problems are witness isomorphic. So we cannot expect to have a single degree under cover-preserving reductions. 
\item[-] Second of all, we expect that 1-in-$k$ SAT and $k$-SAT are \textbf{not} equivalent under overlap preserving reductions, simply because $k$-SAT displays clustering below the phase transition, whereas we expect 1-in-$k$ SAT, $k\geq 3$, not to. Whether our notions of reductions are appropriate enough to capture this difference is an open problem. 
\end{itemize} 

The main weakness of any attempt of connecting the geometry of instances of two problems by reductions is the asymptotic nature of the predictions provided by methods from Statistical Physics. Specifically, such predictions (the existence of a sharp SAT/UNSAT threshold, the existence of one/multiple clusters of solutions, etc.) refer to properties that are true for \textbf{almost all} (rather than all) instances at a given density. In contrast classical reductions (as well as the ones we give) provide guarantees for \textbf{all} instances. 

Secondly, our definitions have not attempted to map (via the reduction) the precise location of the two phase transitions, and there is no reasonable expectation that ordinary reductions preserve this location. A reason is that, for instance, 
$(2+p)$-SAT \cite{2+p:nature} is $NP$-complete for $p>0$, but the location of the phase transition in $(2+p)$-SAT is determined \cite{2+p:nature,2+p-SAT} by the "easy part", at least for $p<0.4$. So, unless the reduction classifies such problems as "easy" ("reducing" an instance $\Phi$ of $(2+p)$-SAT to its 2-SAT component via a notion of reduction that may sometimes err, but is otherwise rather trivial) one cannot simply use ordinary notions of reductions, even restricted as the ones put forward in this paper\footnote{In the journal version of the paper, to be posted on arxiv in the near future, we plan, in fact, to develop such "reductions"}. 

To conclude: to get the notion of reduction appropriate for dealing with phase transitions we might need a notion of reduction that is only correct on one class of instances (and errs on a negligible set of instances at each density). We don't claim that our reductions accomplish this, and only hope that they can function as useful guidelines towards the development of this "correct" notions of reductions. How to do this (and whether it can be done at all) is intriguing. And open. 

\bibliographystyle{eptcs}
\bibliography{/Users/gistrate/Dropbox/texmf/bibtex/bib/bibtheory}

\end{document}